\documentclass[10pt,journal]{IEEEtran}
\usepackage{url}
\usepackage[pdftex]{graphicx}

\newcommand{\pr}[1]{\mathrm{P}\! \left(#1\right)}

\newcommand{\mean}[1]{\mathrm{E}\!\left[#1\right]}

\newcommand{\set}[1]{\left\{{#1}\right\}}

\newcommand{\indic}[1]{1_{\{#1\}}}

\usepackage{bm}
\usepackage{amsmath}
\usepackage{amsthm}
\newtheorem{lemma}{Lemma}
\newtheorem{theorem}{Theorem}

\newtheorem{definition}{Definition}
\newtheorem{remark}{Remark}
\newcommand{\petteri}[1]{}
\newcommand{\eat}[1]{}
\begin{document}
\title{Performance of wireless network coding: motivating small encoding numbers}
\author{Petteri~Mannersalo,~\IEEEmembership{Member,~IEEE,}
        Georgios S. Paschos
        and~Lazaros~Gkatzikis,~\IEEEmembership{Student Member,~IEEE}%
\thanks{P. Mannersalo is with VTT, Finland. G.S. Paschos is with CERTH-ITI and L. Gkatzikis is with University of Thessaly, both in Greece.}}

\maketitle
\begin{abstract}
  This paper focuses on a particular transmission scheme called local network coding,
  which has been reported to provide significant performance gains in
  practical wireless networks. The performance of this scheme strongly
  depends on the network topology and thus on the locations of the
  wireless nodes.
  Also, it has been shown previously that finding the encoding strategy, which achieves maximum performance, requires complex calculations to be undertaken by the wireless node in real-time.
    
  Both deterministic and random point pattern are
  explored and using the Boolean connectivity model we provide upper
  bounds for the maximum coding number, i.e., the number of packets
  that can be combined such that the corresponding receivers are able to decode. For the
  models studied, this upper bound is of order of $\sqrt{N}$, where $N$
  denotes the (mean) number of neighbors. Moreover, achievable
  coding numbers are provided for grid-like networks. We also calculate the multiplicative constants that determine the gain in case of a small network. Building on
  the above results, we provide an analytic expression for the upper
  bound of the efficiency of local network coding.
  The conveyed message is that it is favorable to reduce computational complexity by relying only on small encoding numbers since the resulting expected throughput loss is negligible.
 \end{abstract}
\section{Introduction}

Network coding is an exciting new technique promising to improve the
limits of transferring information in wireless networks. The basic
idea is to combine packets that travel along similar paths in order to
achieve the multicast capacity of networks. The network
coding scheme called \textit{local network coding} was one of the first practical implementations able to
showcase throughput benefits, see COPE in \cite{cope}. The idea of local network coding is to
encode packets belonging to different flows whenever it is possible for these packets to be decoded at the next hop. The simplicity of the idea gave hopes for its efficient application in a real-world wireless router. In
the simple Alice-relay-Bob scenario, the relay XORs outgoing packets
while Alice and Bob use their own packets as keys for decoding. The
whole procedure offers a throughput improvement of 4/3 by eliminating
one unnecessary transmission.

Local network coding has been enhanced with the functionality of
\textit{opportunistic listen}. The wireless terminals are exposed to
information traversing the channel, and \cite{cope} proposed a smart
way to make the best of this inherent broadcast property of the
wireless channel. Particularly, each terminal operates in always-on
mode, overhearing constantly the channel and storing all overheard
packets. The reception of these packets is explicitly announced to an
intermediate node, called the relay, which makes the encoding
decisions. Finally, the relay can arbitrarily combine packets of different flows
as long as the recipients have the necessary keys for decoding. Using
the idea of opportunistic listen, an infinite wheel topology, where
everyone listens to everyone except from the intended receiver, can
benefit by an order of 2 in aggregate throughput by diminishing the
downlink into a single transmission, see \cite{proutiere}.
The wheel is a particular symmetric topology that is expected to appear rarely in real settings. Also, the above calculations take into account that all links have the same
transmission rates, thus it takes the same amount of time to deliver a
native (non-coded) packet or an encoded one. In addition, all possible flows are conveniently assumed to exist. This, however, is not expected to be a frequent setting in a real world network. A natural question reads: what is the expected
throughput gain in an arbitrary wireless ad hoc
network?

The maximum gain does not come at no cost either.
Deciding which packets to group together in an encoded packet is not a
trivial matter as explained in \cite{proutiere}, in ER \cite{er} and
in CLONE \cite{clone}. In the latter case, the medium is assumed to be
lossy, and the goal is to find the optimal pattern of retransmissions
in order to maximize throughput. In the first case, a queue-length
based algorithm is proposed for end-to-end encoding of symmetric
flows (i.e. flows that one's sender is the other's destination and the other way around.). All these decision-making problems are formulated as
follows. Denote $N(f_i)$ the set of nodes in need of a packet belonging to flow $f_i$ and
$H(f_i)$ the set of nodes having it. Then the encoded combination of two packets belonging to flows $f_i$ and $f_j$ can be
decoded successfully if and only if $N(f_i)\subseteq H(f_j)$ and $N(f_j)\subseteq
H(f_i)$. If this condition is true we draw an edge on the \textit{coding
  graph} with vertices all the possible packets. Then finding the
optimal encoding scheme is reduced to finding a minimum clique
partition of the coding graph, a commonly known NP-hard problem,
\cite{er}. Moreover, the same complexity appears when the relay node
makes scheduling decisions, i.e., selecting which packets to serve and
with what combinations. Work related to index coding has shown that this problem can be reduced to the boolean satifyability problem (SAT problem), \cite{chaundry:08}.

Thus a second question arises: what is the loss in throughput gain if
instead of searching over all possible encoded packet combinations, we
restrict our search in combinations of size at most $m$?
 In this paper
we are interested in showing that, for a real ad hoc wireless network,
opportunities for large encoding combinations rarely appear. To show this, we consider regular topologies like grids as well as random ones. We calculate the maximum encoding number in these scenarios in the mean sense and we consider small as well as large networks. To capture the behaviour of large (or dense) networks, we examine the scaling laws of maximum encoding number. 

Scaling laws are of extreme interest for the network community in general. Although they hold asymptotically, they provide valuable insights to the system designers. In this direction,
 the authors in \cite{WG09} study the wireless networks scaling capacity in a Gupta-Kumar way taking into account complex field NC. \cite{JGT07} also examines the use of NC for scaling capacity of wireless networks. They find that NC cannot improve the order of throughput, i.e. the $O\left(\frac{1}{\sqrt{n}}\right)$ law prevails.
\cite{AEMO07} discusses the issue of scaling NC gain in terms of delays while \cite{HR08} identifies the energy benefits of NC both for single multicast session as well as for multiple unicast sessions.
In \cite{GK08}, NC is used instead of power control and the benefits are characterized.
In a similar spirit, \cite{KV09} investigates the use of rate adaptation for wireless networks with intersession NC. Utilizing rate adaptation, it is possible to change the connectivity and increase or decrease the number of neighbors per node. They identify domains of throughput benefits for such case.

The most relevant work in the field is \cite{howmany}. The authors analyze the maximum coding number,
i.e., the maximum number of packets that can be encoded together such
that the corresponding receivers are able to decode. They show that
this number scales with the ratio $\frac{R}{\delta}$ where $\delta$ is
a region outside the communication region and inside the interference
region. Note however, that this work does not yield any geometric
property for the frequency of large combinations since it relies only on specific protocol characteristics. In networks with
small $\delta$, e.g., whenever a hard decoding rule is applied, there
is no bound for the maximum coding number.

In this paper we study the problem from a totally different point of
view, showing that there exist inherent geometric properties bounding
the maximum coding number below a number relative to the population or
density of nodes. Moreover, we apply the Boolean connectivity model
for which $\delta=0$, and thus the previous result does not provide
any bound at all.

We show that the upper bound of the maximum coding number is related
to a convexity property that any valid combination has.  We start by
considering a fixed separation distance network, like a square grid,
and show that in such networks, the maximum coding number is
$O(\sqrt{N})$ and $\Omega(\sqrt[4]{N})$,\footnote{ The symbol $O()$
  denotes that the function is bounded above by some linear function
  of the expression in the brackets whereas $\Omega()$ denotes that
  the expression is bounded from below.} where $N$ is the number of
nodes in transmission range of the relay. This implies that, even in networks with
canonically placed nodes, the maximum coding combination is
line-shaped even though the set of all nodes live in
2-dimensions. Next we study a random network where the locations of
the nodes follow a Poisson point process on the plane. In this case,
the maximum encoding number is found to be bounded in probability by
$O(\lambda^{\frac{1}{2}+\epsilon})$, where $\lambda$ is the node
density and $\epsilon>0$ arbitrary. Finally we consider the case where
the encoder searches for combinations of at most size $m<N$. We show
that the throughput efficiency loss in this case depends on the size
of the network, and for small networks the loss can be
negligible. This way we motivate heuristic algorithms that avoid the
high complexity arising in encoding selection. 
Through extensive simulations we show that all the derived results hold in general even for small networks.

\begin{figure}[tb]
\centering
\includegraphics[width=2.5in]{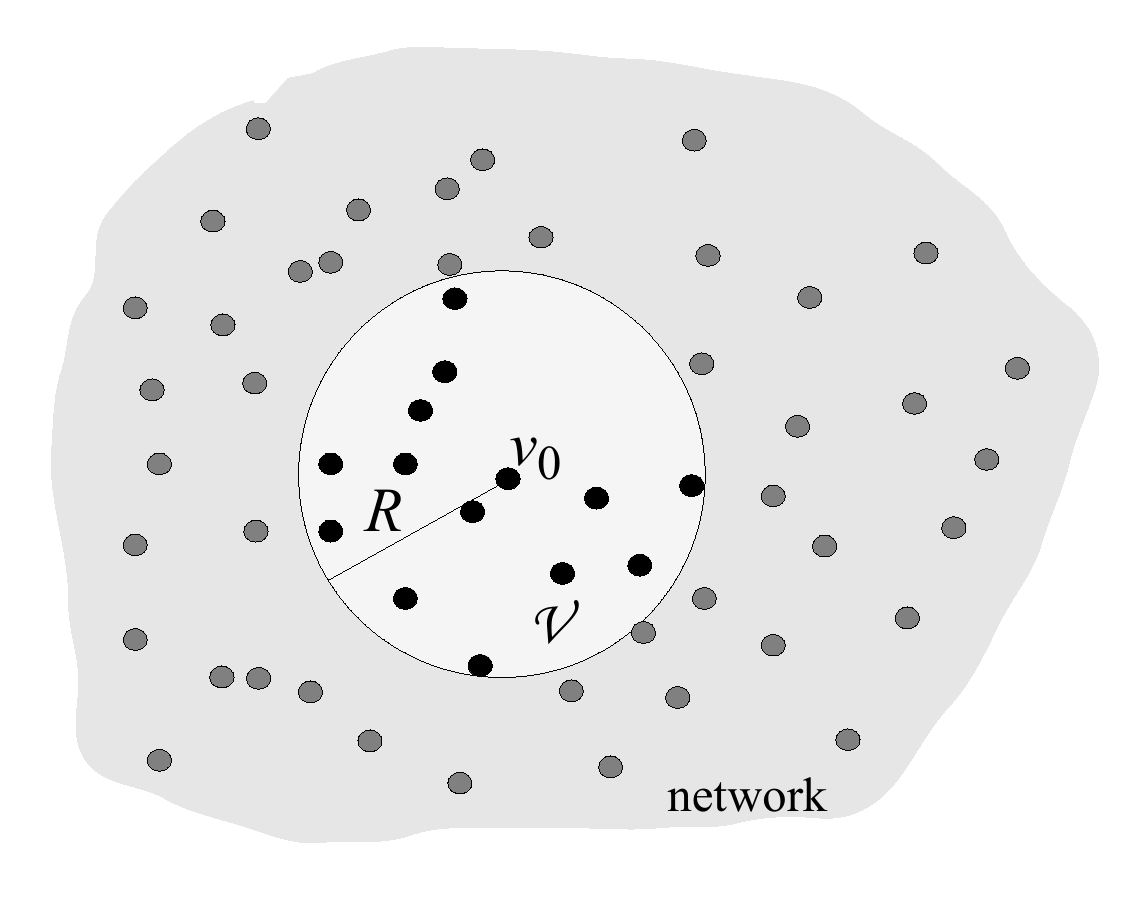} 
\caption{The set of valid nodes $\mathcal{V}$ is selected inside the disk of radius $R$ and origin the location of $v_0$.}
\label{fig:example}
\end{figure}
The paper is organized as follows.  In Section \ref{sec:model}, the
model is described and some basic properties are given.  In Section
\ref{sec:deterministic}, the main results for the case of grid-like
networks are derived.  Then in Section \ref{sec:stochastic} the case
of randomly positioned networks is considered.  A rate analysis is
provided in Section \ref{sec:efficiency} and simulation results are
shown in Section \ref{sec:numerical}.  The paper is concluded in
Section \ref{sec:conclusion}.

\section{Communication model}
\label{sec:model}

We assume a set of nodes $\mathcal{V}$, positioned on the plane.
Communications between these nodes are established via the Boolean
interference model (see, e.g., \cite{franceschetti}). In this model, a
link between two nodes $\{v_i,v_j\}$ is realized if and only if $\left|\mathbf{X}(v_i)-\mathbf{X}(v_j)\right|\leq
R$. In this case, we say that $v_i$ is connected with $v_j$ and vice versa. Note that the Boolean interference model is an undirected graph in the sense that only bi-directional links appear.

\subsection{Information flow}

Each node $v_i$ having degree $\deg\left(v_i\right)>1$, apart from
transmitting and receiving, relays information. In this process, it is
possible to avoid unnecessary transmissions by employing local network
coding. To simplify the analysis, we will consider only one cell,
i.e., we will focus on a given node $v_0$, and all its neighbors, and
calculate the network coding gain on the downlink of this node. A
similar result, then, holds for any such node serving as a relay. Thus we restrict
$\mathcal{V}$ to contain all neighbors of $v_0$, with $\mathcal{V}=\{v_1,v_2,\dots,v_N\}$ and $N\doteq
\left|\mathcal{V}\right|$ the number of nodes under consideration.
For a network determined by a Poisson point process with density
$\lambda$, we use correspondingly the mean number of points which is
given by $\mean{N}=\lambda \pi R^2$.
The main objective of this paper is to find how the \emph{maximum coding number} and the \emph{maximum network coding gain} scale with the number of
nodes. Also we will provide bounds for the scaling constants which are useful for determining the behaviour in small networks.

Apart from the number of neighbors, the gain analysis depends also on
the activated flows. In the simple Alice--relay--Bob topology, it is
possible that only the flow going from Alice to Bob is activated, in
which case the gain is zero. In this paper we are interested in determining
an upper bound for the efficiency loss when the relay is
constrained on combinations of size $m<N$ (e.g. if $m=2$ the system is constrained to pairwise XORing). For this reason, we consider the
maximum gain scenario. For each node designated as a relay, we assume
that all possible two--hop flows traversing this relay are
activated. This means that each node designated as a relay, has all
possible different packets from which to select an XOR combination to
send to the neighbors. Since not all of those combinations are valid,
finding the maximum valid combination that corresponds to the maximum coding number is a non-trivial task and will
be the goal of this paper. The resulting bound will help characterize the efficiency loss due to resorting to $m$-wise encoding. In real systems, some flows might not be active in which case the resulting efficiency loss from $m$-wise encoding will be even smaller.

To make this more precise, similar to \cite{er}, we define
source-destination pairs designating 2--hop flows that cross the
relay. Each flow $f\in\mathcal{F}$ has a source $S(f)$, a destination
$D(f)$, a set of nodes having it $H(f)\subset \mathcal{V}$ (either by
overhearing or ownership) and a set of nodes needing it $N(f)\subset
\mathcal{V}$. We write $\subset$ because at least one node, the destination $D(f)$ or the source $S(f)$, is not part of $H(f)$ and
$N(f)$ correspondingly. Two flows $f_1, f_2$ are called symmetric when they satisfy the property $S(f_1)=D(f_2)$ and $D(f_1)=S(f_2)$.

\subsection{Constraints}

Here we summarize the previous subsection in the form of
constraints. We will focus on network coding opportunities appearing
in the aforementioned arbitrary network around the relay
$v_0$. 
\eat{Without loss of generality, designate $\mathbf{X}\left(v_0\right)$ as the origin.}

\petteri{iff is not used in definitions}

\begin{definition}
  \emph{(valid node):} A node $v_i\in \mathcal{V}$ is a \textit{valid node} if
  $\left|\mathbf{X}(v_i)-\mathbf{X}(v_0)\right|\leq R$.
\end{definition}

\begin{definition}
  \emph{(valid flow):} A flow $f\in \mathcal{F}$ is a \textit{valid
    flow} if $S(f)$ and $D(f)$ are valid nodes not neighboring with
  each other, i.e., $\left|\mathbf{X}(S(f))-\mathbf{X}(D(f))\right|>
  R$.
\end{definition}

\begin{definition}
  \emph{(valid combination):} A subet of flows $\mathcal{C}\subseteq \mathcal{F}$ with $\mathcal{C}=\{f_1,f_2,
  \dots, f_C \}$, where $C=\left|\mathcal{C}\right|$, is a \textit{valid combination} if
\begin{itemize}
\item each flow $f_i\in\mathcal{C}$ is a valid flow,
\item every pair of flows $f_i,f_j\in \mathcal{C}, f_i \neq f_j$ satisfies $N(f_i)\subseteq H(f_j)$ and
  $N(f_j)\subseteq H(f_i)$ or equivalently, $S(f_i)$ is connected with
  $D(f_j)$ while $S(f_j)$ is connected with $D(f_i)$.
\end{itemize}
\end{definition}

We define the maximum coding number $C_{\max}$ as the greatest cardinality among all valid sets $\mathcal{C}$. If the positions of the network are random, $C_{\max}$ is evidently a random variable.

Note that we could
impose additional constraints. For example, if a flow can be routed
more efficiently by a node other than $v_0$, then this flow should be excluded
from the set of valid flows. This would restrict further the set of
valid combinations and thus by omitting this constraint we derive an upper bound for $C_{\max}$.

Next, we state some fundamental properties of the valid combinations\eat{some of which 
arise due to the bidirectional properties of the connectivity model}. For
each flow $f$ belonging to a valid combination $\mathcal{C}$ we have
\begin{itemize}
	\item $D(f)\in N(f)$,
	\item $D(f) \in H(j)$ for all $j\in \mathcal{C}\setminus \{f\}$,
\end{itemize}
which leads us to the following properties.

\begin{remark}
  The destination node of $f_i\in \mathcal{C}$ is different from the
  destination node of any other flow $f_j\in \mathcal{C}\setminus
  \{f_i\}$.
\end{remark}

\begin{remark}
  The source node of $f_i\in \mathcal{C}$ is different from the source
  node of any other flow $f_j\in \mathcal{C}\setminus \{f_i\}$.
\end{remark}

Next, we provide a result on the topology for a valid combination. Let $\mathcal{X}_{\mathcal{C}}$ represent the set of locations of all nodes being the source or destination of a flow belonging to a combination $\mathcal{C}$.

\petteri{Changed the notation and modified the proof to hold arbitrary combination}

\begin{lemma}
  Any valid combination of size 3 or larger corresponds to a convex polygon (the polygon is formed using the set $\mathcal{X}_{\mathcal{C}}$ as edges).
  \label{lemma:convex}
\end{lemma}

\begin{proof}
  Consider a valid combination defined by flows 
  \[
  \mathcal{C}=\set{f_i, i=1,\ldots, C},
  \] 
  where $C\geq 3$. Consider also the set of nodes that are sources and/or destinations in $\mathcal{C}$
  \[
  \mathcal{V}_{\mathcal{C}}=\cup_{i} S(f_i) \cup_{i} D(f_i)
  \]
  and the induced set of locations $\mathcal{X}_{\mathcal{C}}$ such that we have a bijective mapping for each element $v_j\in \mathcal{V}_{\mathcal{C}}$ with an element $\mathbf{X}(v_j) \in \mathcal{X}_{\mathcal{C}}$.
  
  Assume that there is a node $v_j\in
  \mathcal{V}_{\mathcal{C}}$ which is an interior point of the convex
  hull\footnote{The convex hull of points $\mathcal{X}$ is the minimal convex
    set containing $\mathcal{X}$.} of $\mathcal{X}_{\mathcal{C}}$.  Thus its location
  $X_j=\mathbf{X}(v_j)$ can be written as $X_j=\sum_{i\not=j} \alpha_i
  X_i$ where $\sum_{i\not=j} \alpha_i=1$ and $\alpha_i\geq 0$ for all
  $i$.

  On the other hand, there is a unique $v_{j^*}\in V_{\mathcal{C}}$,
  which is the communicating pair (source or destination) of $v_j$ in at least one flow, so that
  \begin{equation}
  |X_j-X_{j^*}|>R. 
  \label{eq:proximity}
  \end{equation}
   
  All the other nodes (destinations or sources) in $\mathcal{V}_{\mathcal{C}}$ should be able to reach the node $v_{j^*}$
  directly. Thus,
  \[
  |X_j-X_{j^*}| \leq \sum_{i\not=j}\alpha_i |X_i-X_{j^*}| \leq
  \sum_{i\not=j}\alpha_i R \leq R,
  \]
  which is a contradiction to (\ref{eq:proximity}). Consequently the node $v_j$, as well as all other nodes
  of the combination, necessarily lie on the perimeter of the
  convex hull. Thus, the nodes of a valid combination are the vertices
  of a convex polygon.
\end{proof}

When the set of sources is identical to the set of destinations, the combination consists of symmetric flows only and $C=2(\left|\mathcal{V'}\right|-1)$,
$\mathcal{V'}\subseteq \mathcal{V}$. 

In order to calculate an upper
bound of the network coding combination size, it is enough to resort
to the case of symmetric flows.

\begin{lemma}
  For any valid combination there exists at least one combination of
  the same or larger size that contains only symmetric flows.
\end{lemma}
\begin{proof}
  We will show that for any flow we can add the symmetric one without invalidating the combination as long as it is not already counted.
  
  In a bipartite graph with all the nodes $\mathcal{V}$ on one side
  and the destinations of $\mathcal{C}$ on the other, consider a
  directional link $\ell_f$, between the source of flow $f$ and its
  destination, for each $f\in \mathcal{C}$. Note now that the nodes
  having out-degree one, i.e., the active sources in $\mathcal{C}$, may
  or may not be identical to one of the destination nodes. We can make
  a partition of the set of active sources by assigning those with
  the above property to the set $\mathcal{T}_{\textsl{sym}}$ and the rest
  to the complementary set $\overline{\mathcal{T}_{\textsl{sym}}}$. If
  $\overline{\mathcal{T}_{\textsl{sym}}}=\emptyset$, then the Lemma is
  proved since $\mathcal{C}$ is a valid combination with symmetric
  flows only. If not, then we can create a new combination $\mathcal{C}'$
  which has more flows than the original one using the following
  process. For each transmitter in
  $\overline{\mathcal{T}_{\textsl{sym}}}$, say $S(f_i)$ the transmitter of flow
  $f_i$, add one extra flow $f_i'$ with $S(f_i')=D(f_i)$ and
  $D(f_i')=S(f_i)$. This flow does not belong to $\mathcal{C}$
  (because $S(f_i)\in \overline{\mathcal{T}_{\textsl{sym}}}$) and it
  does not invalidate the combination due to the bidirectional
  properties of the model. Note that $f_i'$ is a valid flow because
  $S(f_i')$ cannot be connected to $D(f_i')$ due to validity of $f_i$. Note also that $S(f_i')$ is connected to $D(f)$ for all $f\in
  \mathcal{C}$ since this is again required for the decoding of the
  original flows. Thus, for any flow we can add the symmetric one
  without invalidating the combination.
\end{proof}

\begin{remark}
  If a valid combination consists of symmetric flows only, its size must be even.
\end{remark}

In graph theory terms, a valid combination with symmetric flows can be
thought of as a graph created by a clique of $C+1$ nodes, minus a
matching with $\frac{C}{2}$ edges, with all symmetric flows defined by
this matching activated. This graph is called in \cite{cope}
\textit{wheel topology}.


\section{Analysis in grid-like topologies}
\label{sec:deterministic}

In this section we focus on positioning the nodes on a 
grid.  Grid topologies often offer an insightful first step approach
towards the random positioning behaviour. Also, the investigation of grids answers the question whether it is possible to achieve high NC gain by arranging the locations of the nodes.

We therefore assume a
network with the additional property
$\left|\mathbf{X}(v_i)-\mathbf{X}(v_j)\right|\geq d$, for any pair of
nodes $v_i, v_j \in \mathcal{V}$. This condition pertains to regular
grids such as the square, the triangular and the hexagonal grid as
well as other grids with non-uniform geometry. We impose nevertheless
the property that the node density is the same over all cells and thus
the geometry should be somehow homogeneous.  The number of nodes
inside a disk or radius $R$ will be $N=O\left((\frac{R}{d})^2\right)$
for these networks and the corresponding node density
$\lambda^{\mathrm{grid}}=O\left((\frac{1}{d})^2\right)$.

\begin{theorem}
  \emph{(Upper bound)} The maximum coding number in
  fixed-separation networks is $O\left(\sqrt{N}\right)$ where $N$ is
  the number of nodes or equivalently
  $O\left(\sqrt{\lambda^{\mathrm{grid}}}\right)$.
\end{theorem}

\begin{proof}
  From Lemma \ref{lemma:convex} we know that the nodes belonging to
  the maximum combination form a convex polygon. Any such polygon
  fitting inside the disk of radius $R$ must have perimeter smaller
  than $2\pi R$. Since the nodes on the perimeter should be at least
  $d$ away from each other, we conclude that the maximum coding number is
\[
C_{\max}<\frac{2\pi R}{d}.
\]
This combined with $N=O\left(\frac{R}{d}^2\right)$ or respectively $\lambda^{\mathrm{grid}}=O\left((\frac{1}{d})^2\right)$ yields the result.
\end{proof}

A particular case of the above bound is the square grid. The number of
nodes inside the disk is $N=\pi \left(\frac{R}{d}\right)^2
+e\left(\frac{R}{d}\right)$ where $e\left(\frac{R}{d}\right)\leq
2\sqrt{2}\pi \frac{R}{d}$ is an error decreasing linearly with
$d$. Thus we obtain an upper bound
\[
C_{\max}^{\textsl{square}}<\sqrt{4\pi N}.
\]

So far we have shown that any network with fixed separation distance
$d$ and uniform density, will have maximum coding number of
$O(\sqrt{N})$, where $N$ is the number of nodes connected to the
relay. In particular, the constant can be determined for any given grid and for the square grid is $2\sqrt{\pi}$. The simulations show that the actual maximum encoding number is approximately half of that calculated above. The reason for that is basically that the valid polygon is always smaller than the disk of radius $R$ and often close to the size of a disk of radius $\frac{R}{2}$.

It is interesting to bound the achievable maximum coding number
from below as well. To obtain intuition about this bound we start
with a non-homogeneous topology, the \textit{cyclic grid}. We
construct concentric cyclic groups of radius $R_i=id$, $i=0,1,\dots,\lfloor \frac{R}{d} \rfloor $
that fall inside the disk of radius $R$. Each cyclic group has as many
nodes as possible such that the fixed separation distance condition
is not violated. Such a topology exhibits different behavior depending on
the selected origin (it is not homogeneous), nevertheless it helps
identify a particular behavior of the achievable maximum coding number. The
cyclic group at $R_i$ has $\left\lfloor
  \frac{2\pi}{\arccos\left(1-\frac{1}{2i^2}\right)}\right\rfloor$
nodes. Thus the grid of radius $R$ will have
\[
N=1+\sum_{i=1}^{\left\lfloor \frac{R}{d}\right\rfloor} \left\lfloor
  \frac{2\pi}{\arccos\left(1-\frac{1}{2i^2}\right)}\right\rfloor.
\]
A very good approximation is, 
\begin{equation}
\begin{split}
N&\approx 1+6\sum_{i=1}^{\left\lfloor \frac{R}{d}\right\rfloor}i\\
&=1+\left\lfloor \frac{R}{d}\right\rfloor+\left\lfloor \frac{R}{d}\right\rfloor^2\\
&\approx 3\frac{R^2}{d^2}.
\end{split}
\end{equation}

\begin{theorem}
\emph{(Lower bound)}
\label{th:lower}
In networks with nodes $d$ away from each other and cell radius $R$,
an achievable maximum coding number is
\begin{enumerate}
\item $C_{\max}^{\textsl{cyclic}}=\Omega(\sqrt{N})$ for cyclic grid with $\mathrm{rem}\left(
    R,2d\right) \rightarrow 0$,
\item $C_{\max}^{\textsl{cyclic}}=\Omega(\sqrt[4]{N})$ for cyclic grid with $\mathrm{rem}\left(
    R,2d\right) \rightarrow d$ and
\item $C_{\max}^{\textsl{square}}=\Omega(\sqrt[4]{N})$ for the square grid of $d$,
\end{enumerate}
where $\mathrm{rem}(x,y)$ is the remainder of the division $x/y$.
\end{theorem}

\begin{proof}

  \underline{For the cyclic grid when $\mathrm{rem}\left( R,2d\right)
    \rightarrow 0$:} By focusing on the cyclic group with
  $\inf\{i:R_i>\frac{R}{2}\}$, note that each node is
  $\frac{R}{2}+\epsilon$ from the center and thus the desired
  connectivity properties are satisfied for all nodes on the cyclic
  group. In this case, we can calculate the number of nodes in the group as
\[
C_{\max}^{\textsl{cyclic}}=\left\lfloor \frac{2\pi}{\arccos\left(1-\frac{2d^2}{R^2}\right)}\right\rfloor,
\]
which for large $N$ is bounded from below by some linear function of $\sqrt{N}$.

\underline{For the cyclic grid when $\mathrm{rem}\left(R,2d\right)
  \rightarrow d$:} Now each node is $\frac{R}{2}+d-\epsilon$ away from
the center and thus we need to select those nodes satisfying the
property of valid combination. For this it is enough that we leave an
empty angle $\phi$ such that if $AOB$ is a diameter and $AOC$ this
angle, then $CB\leq R$. By solving this for the maximum number of
points satisfying this property we get
\[
C_{\max}^{\textsl{cyclic}}=\left\lfloor \frac{2\pi}{\arccos\left(\frac{R^2}{2\left(\frac{R}{2}+d\right)^2}-1\right)}\right\rfloor,
\]
which for large $N$ is bounded from below by some linear function of $\sqrt[4]{N}$.

\begin{figure}[tb]
\centering
\includegraphics[width=0.6\columnwidth]{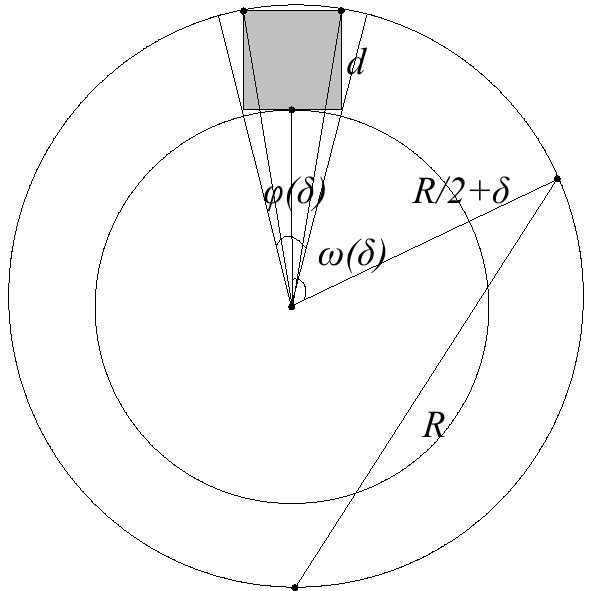} 
\caption{Sketch for the proof of Theorem \ref{th:lower}.}
\label{fig:squarecalc}
\end{figure}

\underline{For the square grid:} we construct a ring around the circle
of $\frac{R}{2}$ radius. The width of the ring is $\delta$ wide enough
to fit a whole squre of dimensions $d \times d$. Such a square is
bound to contain exactly one node of the grid. Using Figure
\ref{fig:squarecalc}, and the triangles relative to the small square,
we calculate $\delta$ as
\[
\delta=\frac{\sqrt{R^2+d(5d+4R)}-R}{2}.
\]
Thus we can show that $d\leq\delta\leq \frac{3d}{2}$. If we use
the largest possible value that guarantees that the ring contains one node
at each step, namely $\delta=1.5d$, we can compute the angle that
contains at least one node, which is of the order of $d$:
\[
\phi(\delta)=\arcsin\left(\frac{\sqrt{2}d}{\sqrt{d^2+R^2}}\right).
\]
Finally, we compute the angle which should be left empty in the valid
combination such that any node outside this angle is reachable by the
most distant node (the one at the bottom).
\[
\omega(\delta)=\arcsin\left(\frac{R}{\sqrt{2}(R+1.5d)}\right).
\]
This angle is of the order of $\sqrt{d}$. Finally an achievable
combination is obtained if we alternate $\phi$ and $\omega$ until we
fill the circle.
\[
C_{\max}^{\textsl{square}}=\left\lfloor \frac{2\pi}{\phi(\delta)+\omega(\delta)}\right\rfloor,
\]
which is bounded from below by a linear function of
$\sqrt[4]{N}$. Note that the sparseness of the combination is due to
$\omega(\delta)$ and a possible reasoning is that the bound is
constructed to cover all the cases, thus also the case that the
uncomfortable positioning of nodes matches the second case of the
cyclic grid above.
\end{proof}

In \cite{alon}, relative results on convex polygons in constrained sets guarantee the existence of convex polygons of size $\Omega(\sqrt[4]{N})$ when the $N$ thrown nodes are kept seperated by some distance.


\section{Stochastic analysis}
\label{sec:stochastic}

Assume that the locations of the nodes are determined by a Poisson
point process with density $\lambda$. The connectivity properties of
this model are well studied in the literature (see
e.g.~\cite{meester,DousseThiranHassler:02}). In our work, we assume
that the network is percolating, i.e. the nodes are dense enough to
ensure multihop communications. As in the deterministic case, we
assume that a relay is located at the origin. For a Poisson point
process this assumption does not change the distribution of the other
points. The main result is an upper bound in probability for the
maximum coding number.

\begin{theorem}
  In a random network determined by a Poisson point process with
  density $\lambda$, the maximum coding number corresponding to combinations having the relay at the origin satisfies
  \[
  \lim_{\lambda\to\infty}\pr{ C_{\max}
    (\lambda)=O(\lambda^{1/2+\epsilon})}=1,
  \]
  for any $\epsilon>0$.
\end{theorem}

\begin{proof}
  Cover the disk of radius $R$ around the origin by disjoint boxes of
  size $\frac{1}{\sqrt{\lambda}}\times \frac{1}{\sqrt{\lambda}}$.  The
  number of nodes inside the boxes is denoted by $N_i$,
  $i=1,\ldots,n(\lambda)$. The $N_i$ are identically and independently
  $\mathrm{Poisson}(1)$ distributed and thus there is a sequence $I_n$
  such that
  \[
  \lim_{n\to\infty}\pr{\max_{1\leq i \leq n} N_i = \mbox{$I_n$ or
      $I_n+1$}} =1,
  \] 
  where $I_n = O\left( \frac{\log n}{\log\log n}\right)$ (see
  \cite{Anderson:70,Kimber:83}). Since $n(\lambda)=O(\lambda)$,
  \begin{equation}
    \lim_{\lambda\to\infty}\pr{\max_{i=1,\ldots,n(\lambda)}
      N_i=O\left(\frac{\log \lambda}{\log\log \lambda}\right)}=1.
    \label{eq:Nlimit}
  \end{equation}

  Next consider a valid combination. By Lemma \ref{lemma:convex},
  the nodes of the combination form a convex polygon.  The perimeter
  of any convex polygon is at most $2 \pi R$ because it is located
  inside the disk of radius $R$. Since at most $O(\sqrt{\lambda})$
  boxes of size $\frac{1}{\sqrt{\lambda}}\times
  \frac{1}{\sqrt{\lambda}}$ are needed to cover the perimeter of any
  convex polygon,
  \[ C_{\max}(\lambda) \leq \max_{i=1,\ldots,n(\lambda)}N_i
  O(\sqrt{\lambda}) \quad \mbox{a.s.}
  \]
  This implies that
  \[
  \pr{C_{\max}(\lambda) \leq O(\sqrt{\lambda} g(\lambda))} \geq \pr{
    \max_{i=1,\ldots,n(\lambda)} N_i \leq O\left(g(\lambda)\right)}.
  \]
  Setting $g(\lambda)=\log \lambda /  \log\log \lambda$, applying
  equation (\ref{eq:Nlimit}) and finally noticing that $\log \lambda
  /\log\log \lambda = O\left(\lambda^\epsilon\right)$ for any
  $\epsilon>0$ completes the proof.
\end{proof}


\section{Rate efficiency of a network coding combination}
\label{sec:efficiency}

We will focus on the downlink of a valid combination of size
$C$. Without loss of generality, assume that the rate vector
$\mathbf{r}=\{r_i\}_{i=1,2,\dots,C}$ is ordered, i.e.,
$r_1<r_2<\dots<r_C$, and that the flow set is permuted accordingly so
that over the link $\left(v_0,D(f_i)\right)$ packets are transfered at a rate
$r_i$.

The data rate is computed as the number of packets of size $P$
transmitted in a virtual frame over the time needed for these
transmissions. Since an encoded packet is always transmitted at the
lowest rate decodable by all receivers and assuming max-min fair
allocation\footnote{This condition of fairness provides the best
  network coding opportunities and it is usually the balance point
  where network coding gain is computed in multiclass networks.}, we
can deduce the maximum throughput rate with network coding as
\[
r_{\textsl{NC}}(C)=\frac{CP}{\frac{P}{\min\{\mathbf{r}\}}}=Cr_1.
\]
The rate without NC would be
\[
r_{\textsl{w}}(C)=\frac{CP}{\frac{P}{r_1}+\frac{P}{r_2}+\dots+
  \frac{P}{r_C}}=C\left(\sum_{i=1}^{C}\frac{1}{r_i}\right)^{-1}=r_h,
\]
where $r_h$ is the harmonic mean of $\mathbf{r}$. Choosing any $m\leq
C$ and allowing for combinations of size $m$ at most, it is easy
to see that if the criteria for valid combinations are fulfilled for
combination of size $C$ then they are fulfilled for all
subsets. The corresponding achievable rate is 
\[
  \begin{split}
    r_m(C)&=\frac{CP}
    {\sum_{i=1}^{\left\lfloor \frac{C}{m} \right\rfloor}\frac{P}{r_{m(i-1)+1}}+\indic{\mathrm{rem}\left(C,m\right)>0}\frac{P}{r_{C-\mathrm{rem}(C,m)+1}}}\\
    &=C\left(\sum_{i=1}^{\left\lfloor \frac{C}{m}
        \right\rfloor+\indic{\mathrm{rem}\left(C,m\right)>0}}\frac{1}{r_{m(i-1)+1}}\right)^{-1},
\end{split}
\]

Next, we derive the network coding gain for the maximum combination
($C$) and for the constrained group ($m\leq C$).
\begin{equation*}
  g(C) \doteq \frac{r_{\textsl{NC}}(C)}{r_{\textsl{w}}(C)} = C \frac{r_1}{r_h}.
\end{equation*}
Note that the gain is a linear function of $C$ and depends on the
particularities of the rate vector. Also,
\begin{equation*}
\begin{split}
  g_m(C) & \doteq \frac{r_m(C)}{r_{\textsl{w}}(C)} = \\
  &=\frac{C}{r_h} \left(\sum_{i=1}^{\left\lfloor \frac{C}{m}
      \right\rfloor+\indic{\mathrm{rem}\left(C,m\right)>0}}
    \frac{1}{r_{m(i-1)+1}}\right)^{-1}\\
  & \geq \frac{C}{r_h} \left(\frac{\left\lceil
        \frac{C}{m}\right\rceil}{r_1}\right)^{-1} \geq (m-1)
  \frac{r_1}{r_h},
\end{split}
\end{equation*}
where in the first inequality we have used that $r_1$ is the minimum
rate, and in the second we have used $m-1 \leq \frac{C}{\left\lceil
    \frac{C}{m} \right\rceil} \leq m$. If we choose equal rates, then
we readily get $g(C)=C$ and $\max\left\{g_m(C)\right\}=m$ as the
maximum gain for both. Finally we can symbolize that $g(C)=\Theta(C)$
and moreover the difference $g(C)-g_m(C)=O(C-m)$. 
Therefore, the efficiency loss is of the order of $\sqrt{N}-m$ which means that for carefully chosen $m$, the loss can be kept small.


\begin{figure}[t]
\centering
\includegraphics[width=0.49\columnwidth]{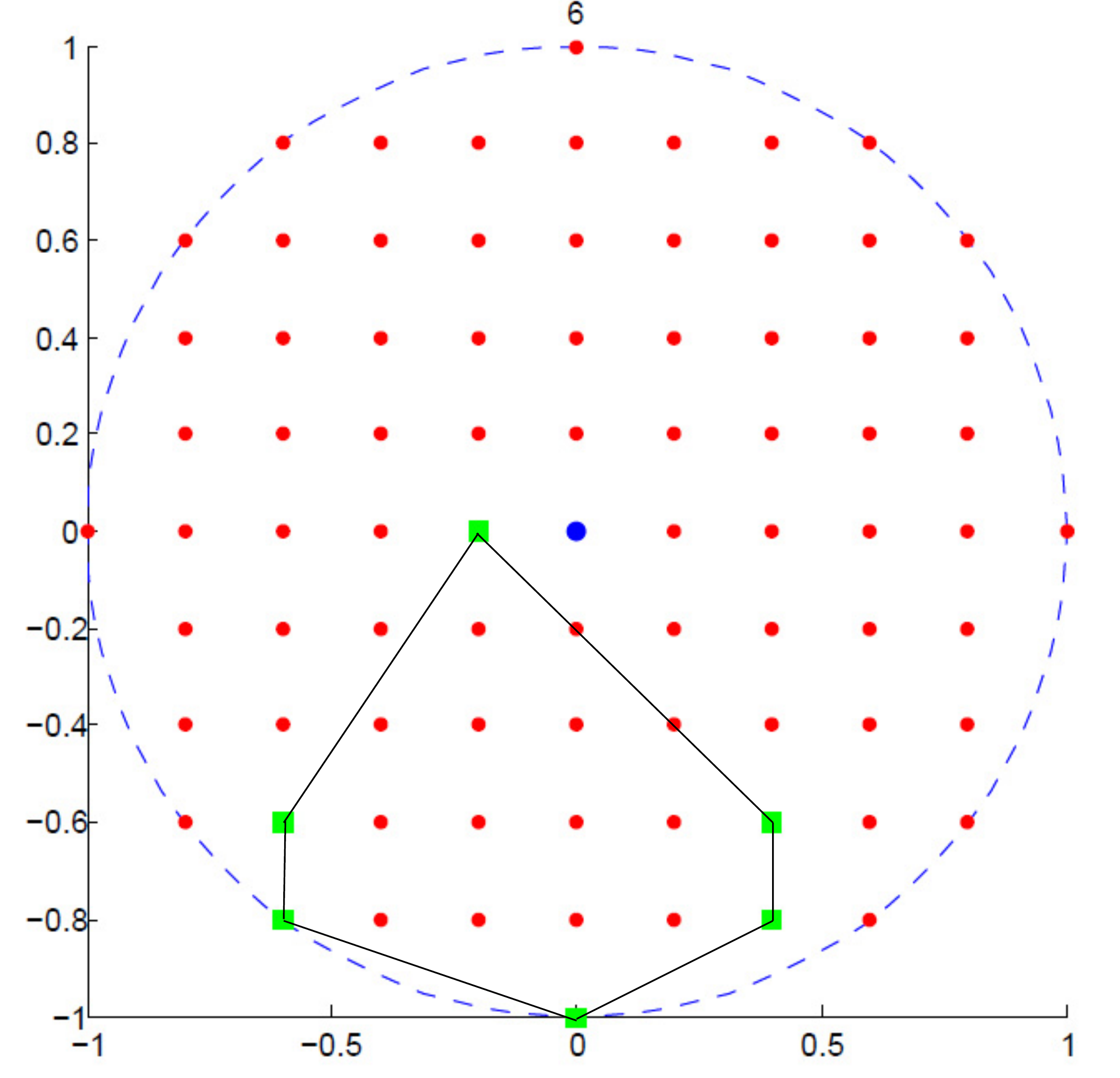}
\includegraphics[width=0.49\columnwidth]{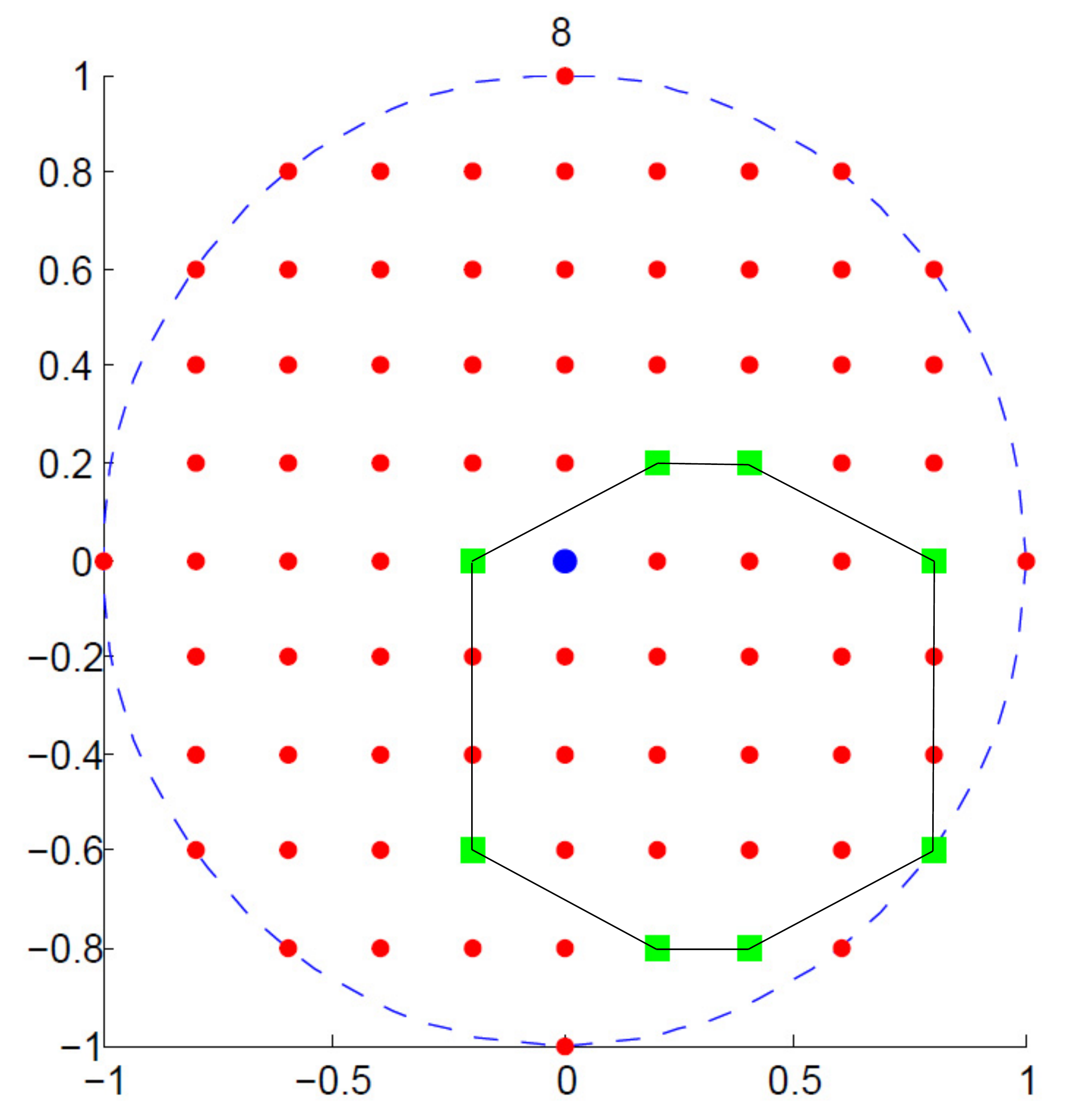}
\caption{Coding combination examples for $C_{\max}=6$ and $C_{\max}=8$
  in a square grid with $N=81$.}
\label{fig:grid_examples}
\end{figure}

\begin{figure}[t]
\centering
\includegraphics[width=0.9\columnwidth]{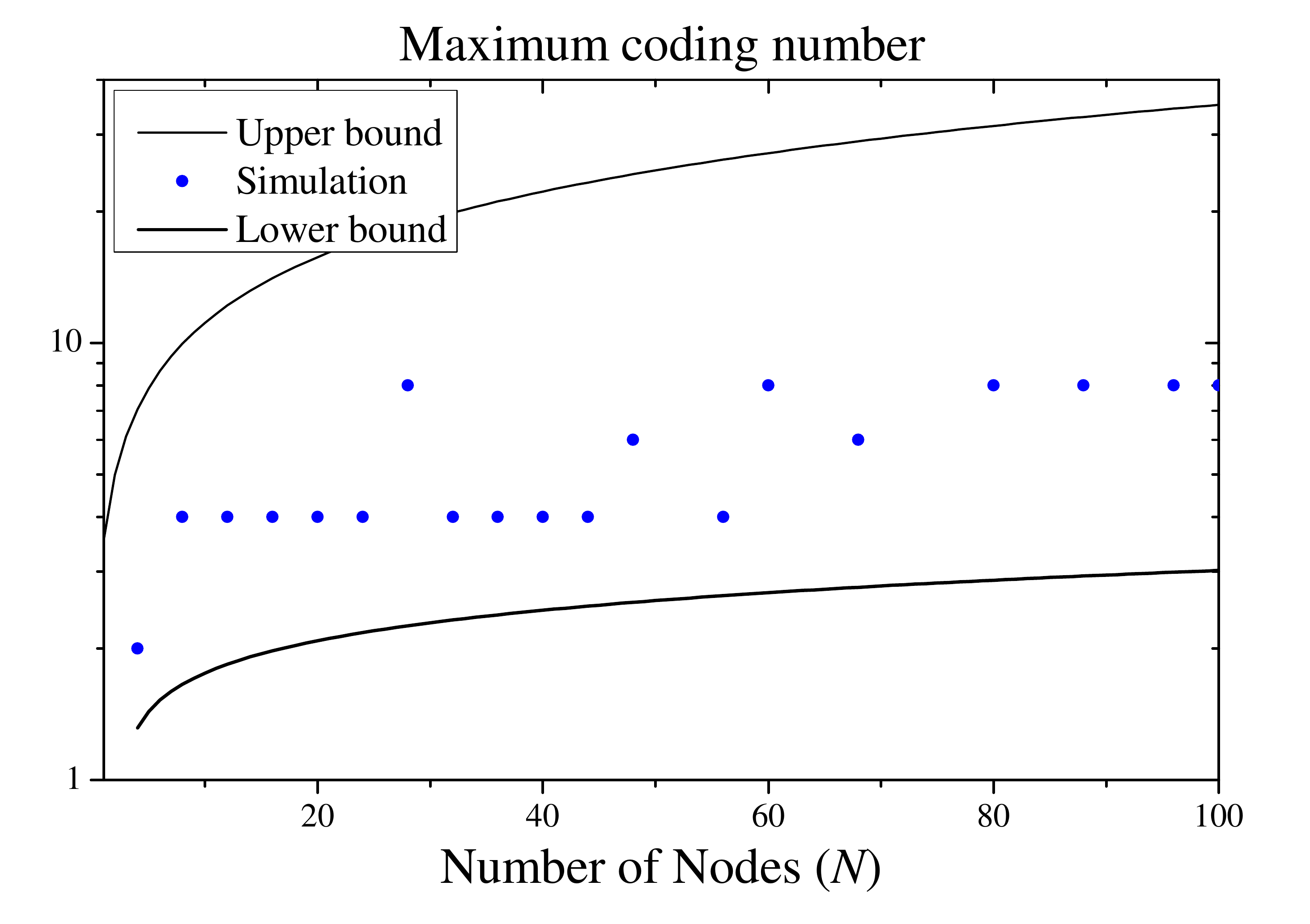}
\caption{Simulation of maximum coding number of a square grid
  inside a disk with $N$ nodes. $C_{\max}$ is the upper bound and
  $C_{\min}$ is the achievable lower bound provided in the theoretical
  analysis.}
\label{fig:grid_maxcomb}
\end{figure}

\section{Numerical results}
\label{sec:numerical}

In this section we present some simulation results that provide
further evidence and insight for our work. For simulation purposes we
consider a disk of radius $R=1$ and a node $v_0$ serving as a relay
situated at the center of the disk. Initially we consider a square
grid of nodes over this disk and we investigate the maximum coding
number, i.e., a set of nodes that satisfies the constraints of section
\ref{sec:model}. Then, the scenario of uniformly random thrown nodes
is considered.

\subsection{Experiments with square grids}

\begin{figure}[tb]
\centering
\includegraphics[width=0.46\columnwidth]{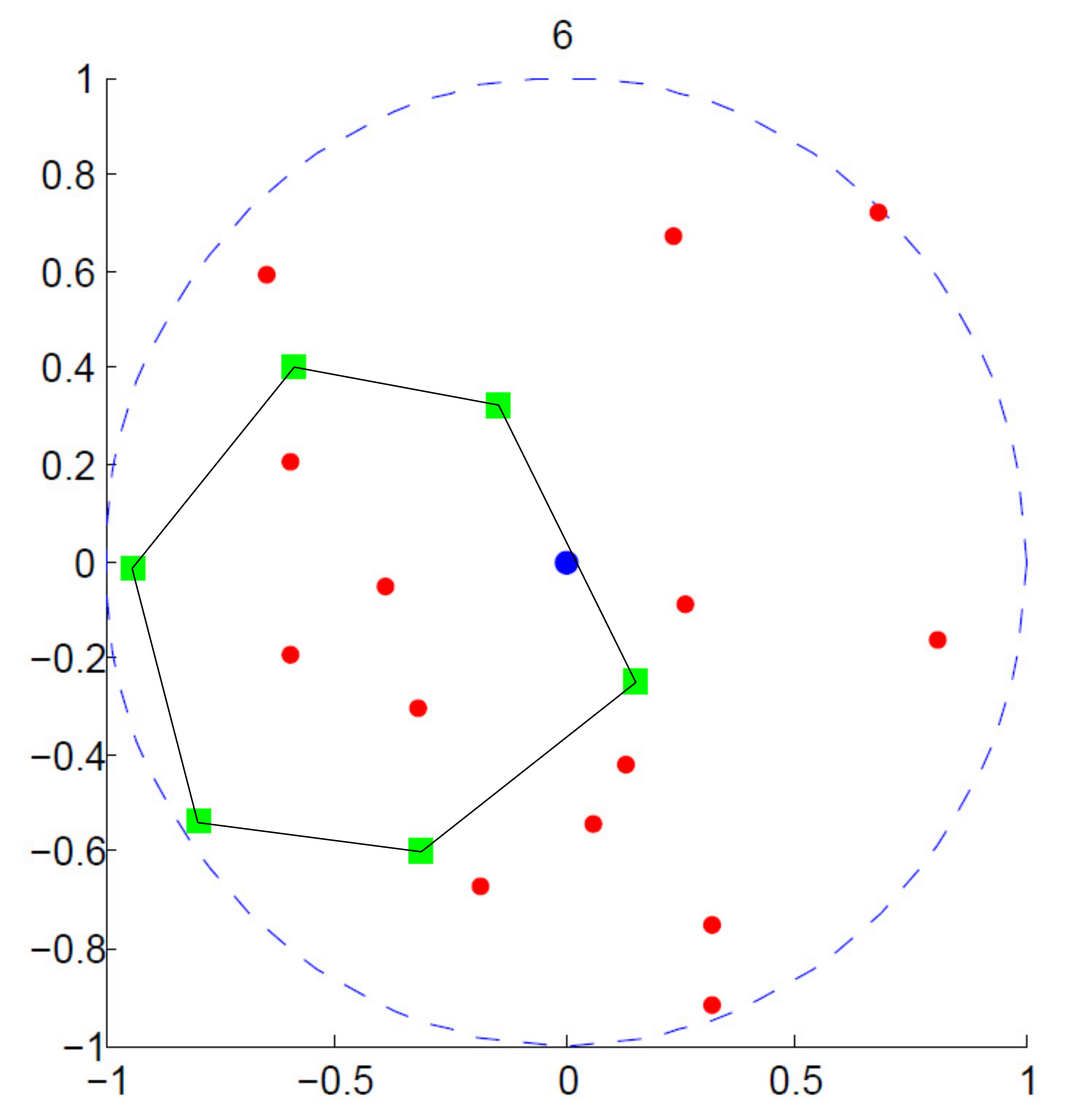}
\includegraphics[width=0.46\columnwidth]{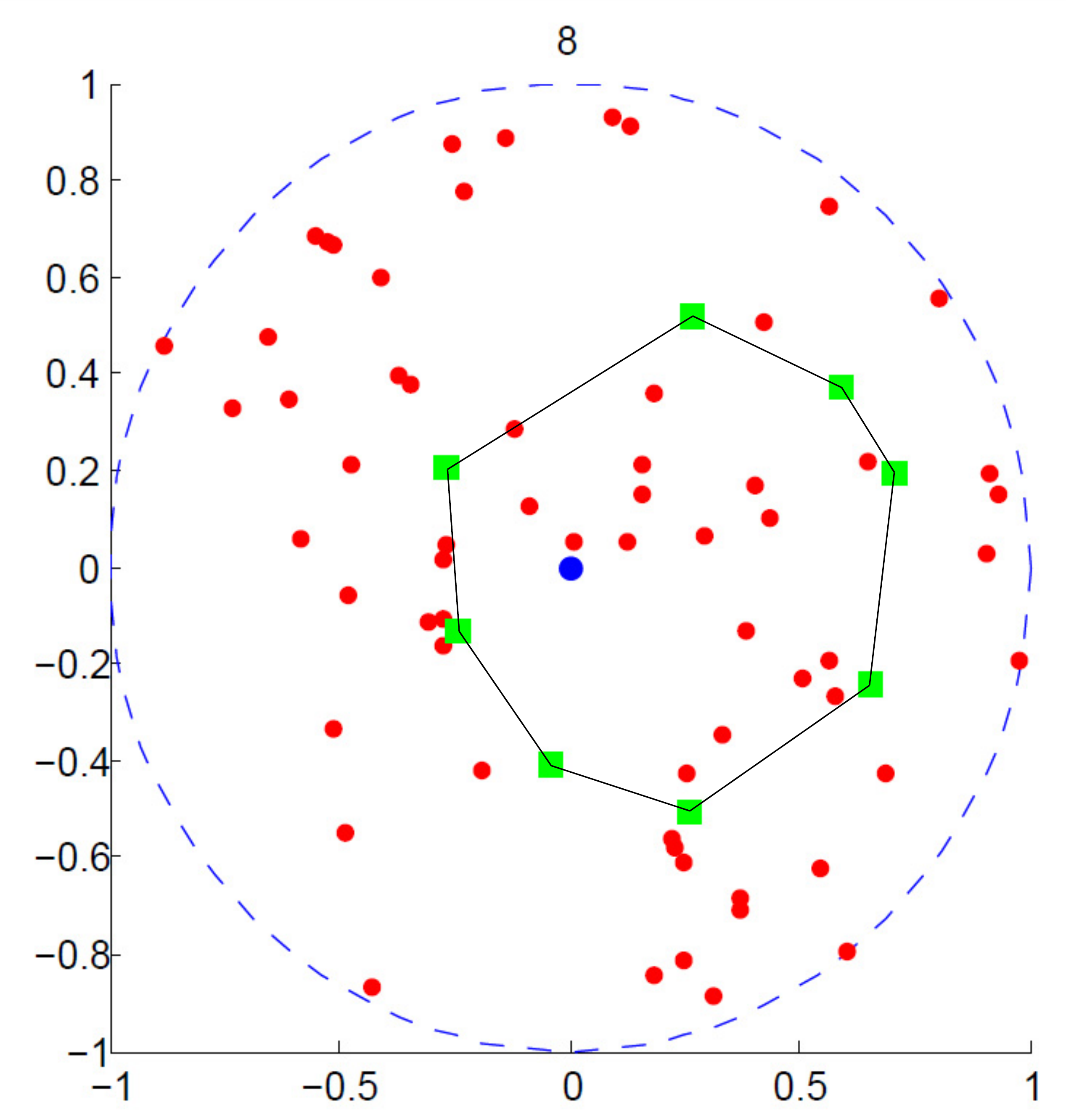}
\caption{Maximum coding combination examples for $C_{\max}=6$ and $C_{\max}=8$ in a uniformly thrown network with $N=21$ and $N=70$ respectively.}
\label{fig:random_examples}
\end{figure}

\begin{figure}[tb]
\centering
\includegraphics[width=0.95\columnwidth]{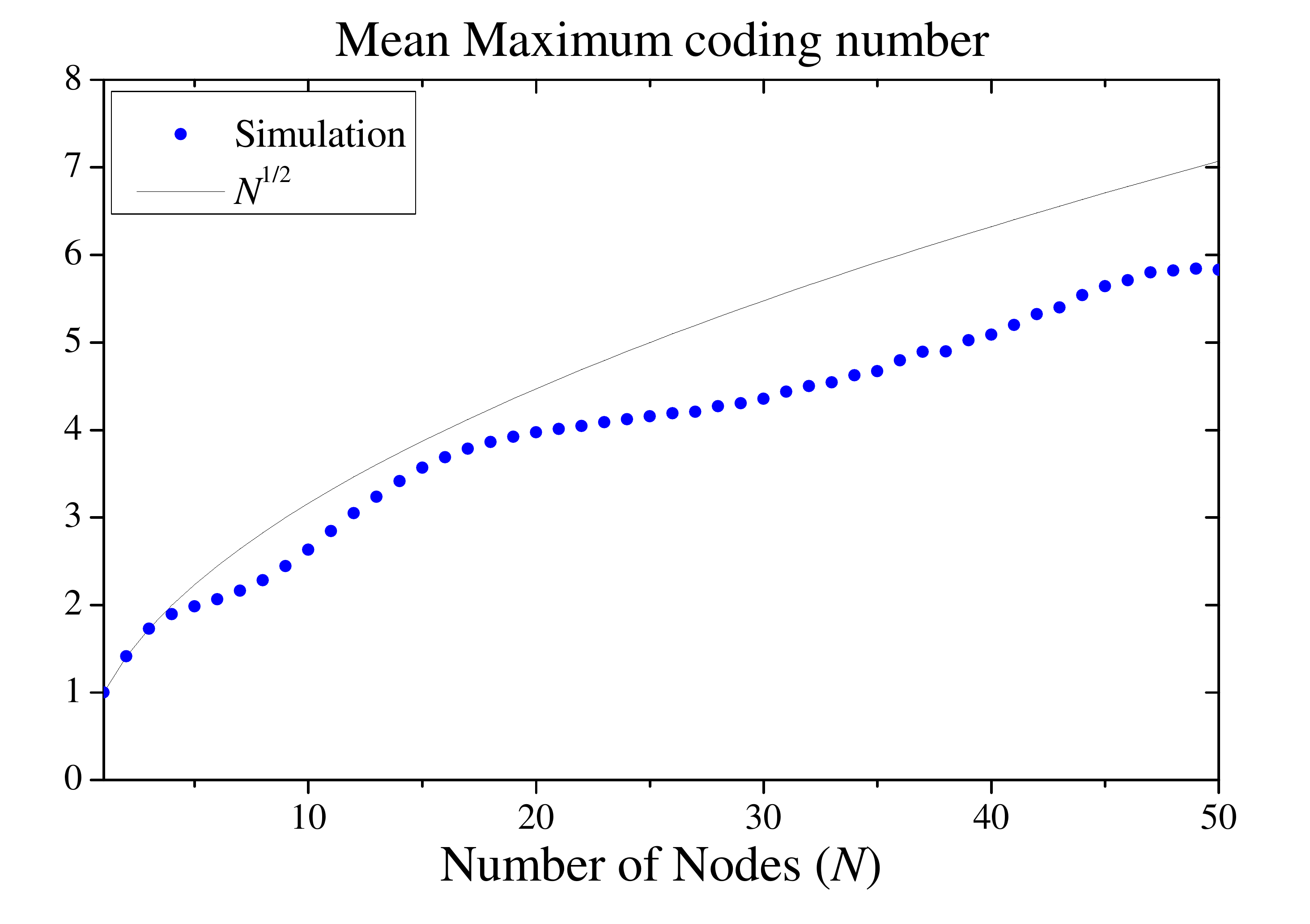}
\caption{Mean maximum coding number in a network of $N$ uniformly thrown nodes.}
\label{fig:random_maxcomb}
\end{figure}

Figure \ref{fig:grid_examples} showcases examples of combination size
$C=6$ and $C=8$ (the maximum is $8$ in this case). During these
experiments we identified an interesting property. We noticed that the
maximum coding number depends only on the number of nodes
inside the disk and not on the actual $d$ used. In particular, for all
$\{d: N \text{ is fixed}\}$, $C_{\max}$ is constant.  In Figure
\ref{fig:grid_maxcomb} we present the results for the square grid. The
actual values seem to be closer to the lower bound than to the upper,
leading to an order closer to $\sqrt[4]{N}$ (notice the logarithmic
scale of the figure). Nevertheless, the oscillating effect due to the 
interplay between the radius and the number of nodes is evident.

\subsection{Experiments with randomly positioned nodes}

Next, we throw $N$ nodes uniformly inside the disk of radius $R=1$.
Examples of maximum coding number are showcased in Figure
\ref{fig:random_examples}. It is noted from these examples that large
combinations tend to appear in a $\delta$--ring form where the inner
side of the ring is a disk of radius $\frac{R}{2}$ and the outer side
is a disk of radius $\frac{R}{2}+\delta$.

In Figure \ref{fig:random_maxcomb}, we present the mean maximum
coding number, for different number of nodes $N$. In each sample,
the maximum
coding number is calculated and the mean is obtained by
averaging over $1000$ random samples. The
$O(\sqrt{N})$ behavior is depicted in this picture.

Figure \ref{fig:random_prob} shows the probability of existence of at
least one coding combination of size $C$ in a network of $N$
uniformly thrown nodes. For example, the maximum component size for
$20\leq N\leq 50$ is either 4 or 6 in the majority of cases. The simulation results show that
in real networks of moderate size the usual combination size is quite
small. The multiplicative constant seems to be close to 1. In this context, the focus should be on developing efficient
algorithms that opportunistically exploit local network coding over a
wide span of topologies using small XOR combinations rather than
attempting to solve complex combinatorial problems in order to find
the best combinations available.

\begin{figure}[tb]
\centering
\includegraphics[width=0.95\columnwidth]{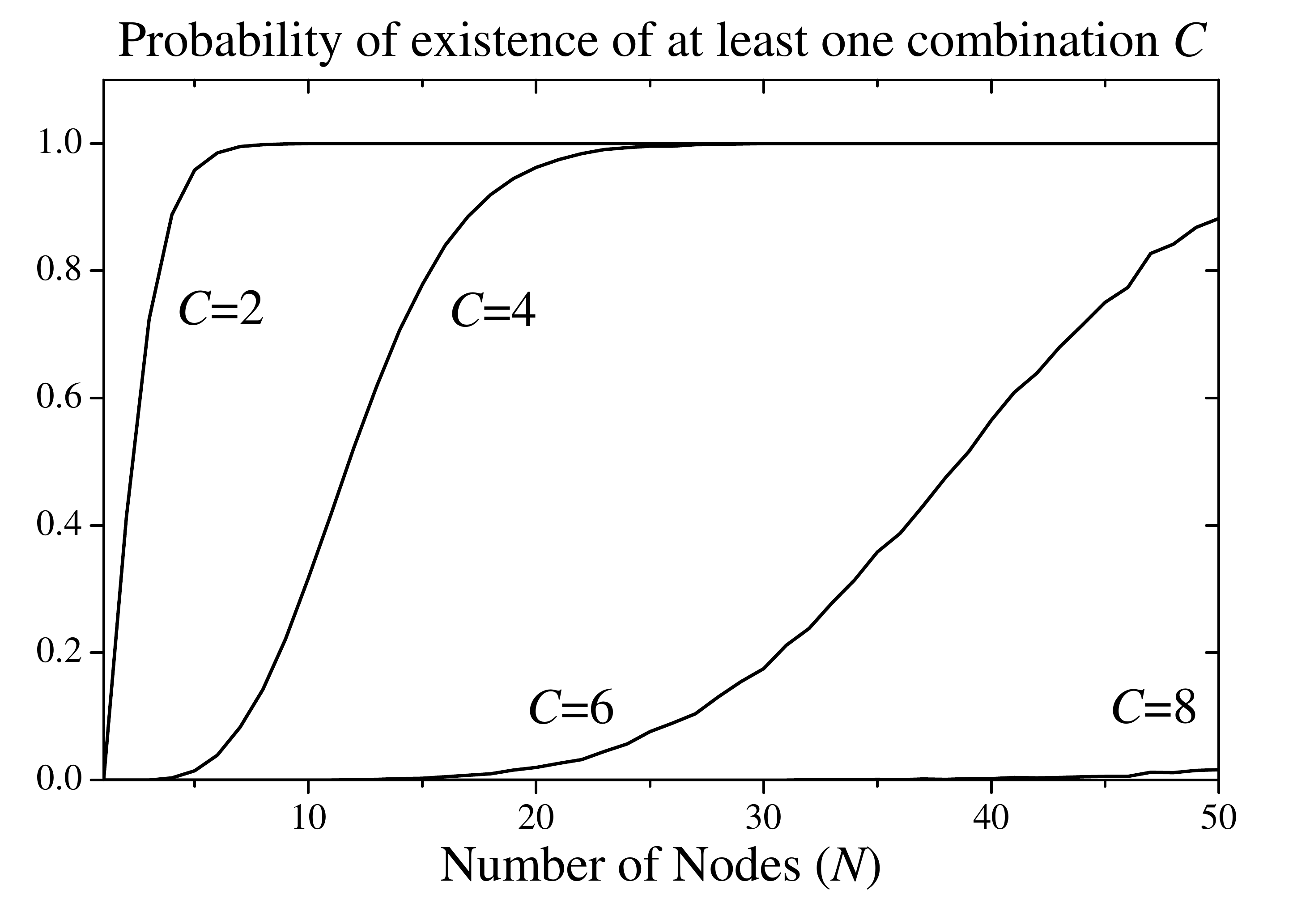}
\caption{Probability of existence of at least one coding combination
  of size $C$ in a network of $N$ uniformly thrown nodes.}
\label{fig:random_prob}
\end{figure}

\subsection{A realistic scenario}
Next, we set up a realistic experiment to be run in simulation environment. A relay node is positioned at the origin, willing to forward any traffic required and apply NC if beneficial. Then, we throw $\frac{N}{2}$ pairs of nodes randomly inside the disk defined by the relay and the communication distance $R$; note that all nodes are valid nodes. Each pair constitutes a symmetric flow. Each flow may be valid or invalid depending on the distance between the two nodes, see definition in section \ref{sec:model}. Whenever the flow is invalid, the nodes communicate directly by exchanging two packets over two slots (one for each direction). If the flow is valid, the relay is utilized to form a 2-hop linear network. Again, 2 packets are uploaded towards the relay using two slots while the downlink part is left for the end of the frame. Finally, the relay has collected a number of packets which may be combined in several ways using NC. To identify the minimum number of slots required to transmit those packets to the intended receivers, we solve the problem of minimum clique partition with the constraint of using cliques of size up to $\frac{m}{2}$ (equivalently, combining up to $m$ packets together). In the above we have assumed that all links have equal transmission rates and that the arrival rates of the flows are all equal (symmetric fair point of operation). The network coding gain is calculated by dividing the number of slots used without NC by the number of slots using NC.

Figure \ref{fig:realistic} depicts results from simulated random experiments. Evidently, it is enough to combine up to two packets per time in order to enjoy approximately the maximum NC gain. This example supports the intuition that in practice the network coding gain from large combinations is expected to be negligible.

\begin{figure}[t]
\centering
\includegraphics[width=0.9\columnwidth]{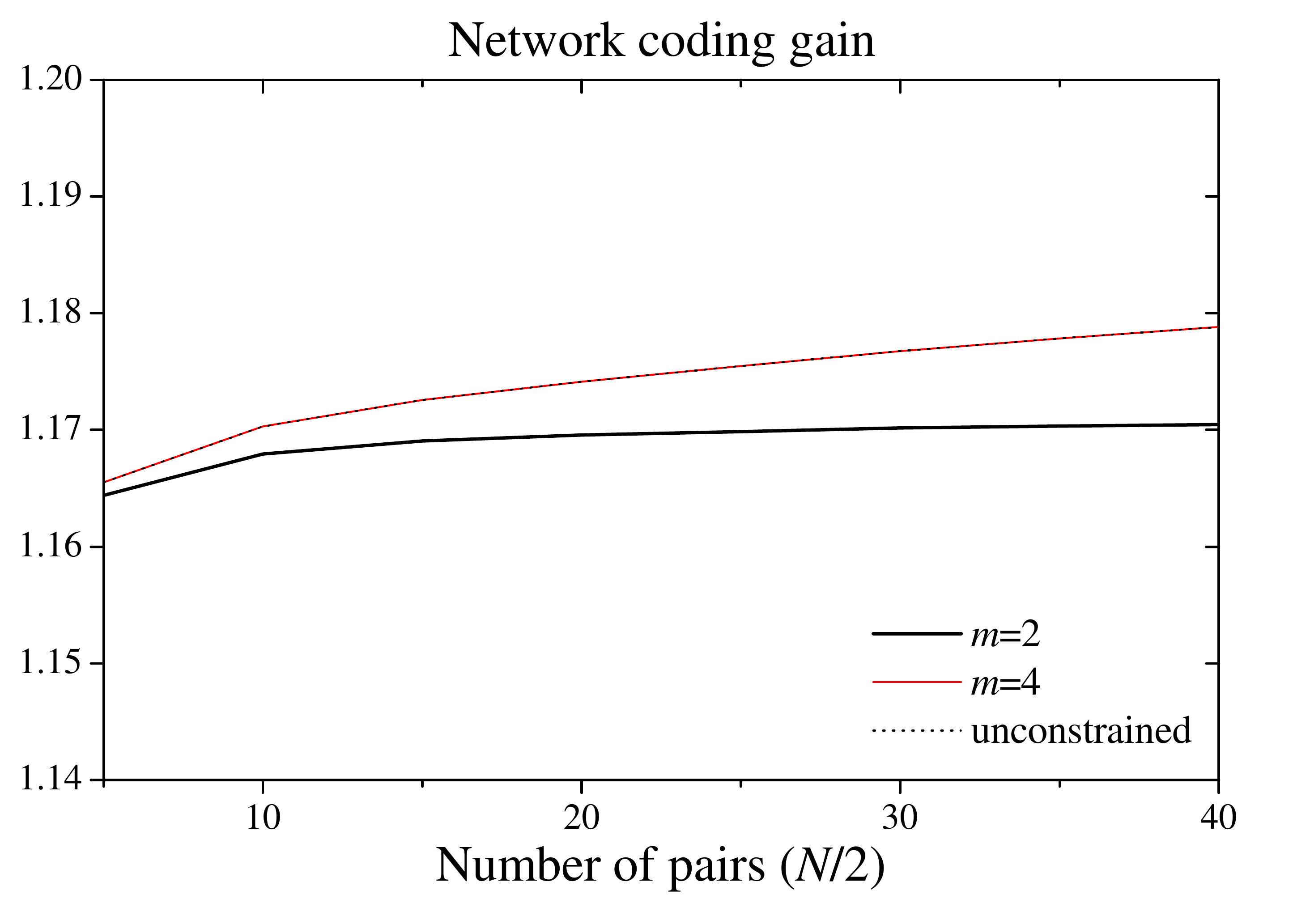}
\caption{A realistic scenario. Pairs of nodes forming flows are randomly thrown inside the disk $(0,r)$. The relay node (situated at the origin) assists the flows that cannot communicate directly. We restrict NC combinations to size $m$ where $m=2,4,\infty$.}
\label{fig:realistic}
\end{figure}
\section{Conclusion}
\label{sec:conclusion}
By considering the Boolean connectivity model and applying the basic
properties required for correct decoding, we showed that for the local
network coding there are certain geometric constraints bounding the
maximum number of packets that can be encoded together. Particularly,
due to the convexity of any valid combination, the sizes of
combinations are at most order of $\sqrt{N}$ for all studied network
topologies.

The fact that the number of packets is limited gives rise to
approximate algorithms for local coding.  Instead of attempting to solve
the hard problem of calculating all possible coding combinations, we
showed that an algorithm considering smaller combinations does not
lose too much.

%
\end{document}